\documentclass[nofootinbib, twocolumn, pra, floatfix]{revtex4}

\usepackage{amsmath}
\usepackage{amsthm}
\usepackage{algorithm}
\usepackage{algorithmic}
\usepackage{amssymb}
\usepackage{graphicx}
\usepackage{color}
\usepackage{bbm}
\usepackage{changes}
\usepackage[hidelinks]{hyperref}

\DeclareMathOperator{\Haf}{\text{Haf}}
\DeclareMathOperator{\expec}{\mathbb{E}}
\DeclareMathOperator{\sech}{sech}

\newcommand{\beq}{\begin{equation}}
\newcommand{\eeq}{\end{equation}}

\newtheorem{thm}{Theorem}

\newtheorem{defn}{Definition}

\begin{document}
\title{Quantum approximate optimization with Gaussian boson sampling}
\author{Juan Miguel Arrazola}
\email{juanmiguel@xanadu.ai}
\author{Thomas R. Bromley}
\email{tom@xanadu.ai}
\author{Patrick Rebentrost}
\email{pr@patrickre.com}
\affiliation{Xanadu, 372 Richmond Street W, Toronto, Ontario M5V 1X6, Canada}

\begin{abstract}
Hard optimization problems are often approached by finding approximate solutions. Here, we highlight the concept of proportional sampling and discuss how it can be used to improve the performance of stochastic algorithms for optimization. We introduce an NP-Hard problem called Max-Haf and show that Gaussian boson sampling (GBS) can be used to enhance any stochastic algorithm for this problem. These results are applied by enhancing the random search, simulated annealing, and greedy algorithms.
With numerical simulations, we confirm that all algorithms are improved when employing GBS, and that GBS-enhanced random search performs the best despite being the one with the simplest underlying classical routine.
\end{abstract}
\maketitle
In this unique time in history, we are witnessing the emergence of the first generation of quantum computers. The so-called noisy intermediate-scale quantum (NISQ) era \cite{preskill2018NISQ} signals the first time that quantum devices capable of outperforming classical computers at specific tasks will become available for public use. With universal fault-tolerant quantum computing still a long road ahead, short-term focus has shifted to identifying the problems that current quantum devices can solve more efficiently than classical computers, without necessarily placing emphasis on their practical applications \cite{boixo2016characterizing, aaronson2016complexity, farhi2016quantum, bremner2016average, bremner2017achieving, douce2017continuous, arrazola2017quantum, gao2017quantum, bermejo2017architectures, harrow2017quantum}.
 
One such example is boson sampling: the task of generating random outcomes from the probability distribution induced by a collection of indistinguishable photons passing through a linear optics network \cite{aaronson2011computational}. Boson sampling was initially conceived as an experimentally appealing method of challenging the extended Church-Turing thesis, which states that all physically realizable processes can be efficiently simulated on a classical Turing machine. Since then, there has been significant effort to implement boson sampling \cite{spring1231692, broome794, tillmann2013experimental, crespi2013integrated} and to propose related models such as scattershot boson sampling \cite{lund2014boson,bentivegna2015experimental,latmiral2016towards} and Gaussian boson sampling \cite{hamilton2017gaussian, kruse2018detailed} that are more amenable to experiments. In all these cases, the focus has been to outperform classical computers at the corresponding sampling task instead of using these devices for solving practical problems. Arguably the first clue that boson sampling could actually be linked to problems of interest came from Ref. \cite{huh2015boson}, where it was shown that Gaussian boson sampling could be used to efficiently infer the vibronic spectra of complex molecules. Proof-of-principle experiments have also been recently reported \cite{clements2017experimental, sparrow2018simulating}.

In this work, which is a companion paper to~\cite{arrazola2018quantum2_published}, we show that Gaussian boson sampling (GBS) can be used to enhance stochastic algorithms for an NP-Hard optimization problem related to finding optimal submatrices. The main insight behind our results is that there exists problems where the GBS distribution naturally assigns a probability to each outcome that is proportional to their value with respect to the underlying figure of merit. This causes good outcomes to be sampled with high probability and bad results to almost never be observed. In cooperation with classical stochastic optimization algorithms, this intrinsic bias can be harnessed to give rise to hybrid algorithms that perform improved searches over the optimization landscape. 

In the following, we discuss the concept of proportional sampling and outline its advantage with respect to uniform sampling in stochastic optimization. We introduce a generic optimization problem which we call the Max-Haf problem and show that the output distribution from Gaussian boson sampling (GBS) is a proportional distribution for this task. Max-Haf is the canonical optimization problem for GBS and it serves as the template for extending GBS to other optimization problems, as we do for the densest $k$-subgraph problem in~\cite{arrazola2018quantum2_published}. We give analytical results quantifying the expected improvement to random search obtained by sampling from the GBS distribution and suggest how GBS can be used within more advanced heuristic algorithms. Our results are complemented with a numerical study of the performance of these enhanced algorithms on an example Max-Haf problem.

\section{Proportional sampling}
A combinatorial optimization problem can be cast as follows. Given a set $X=\{x_1,x_2,\ldots,x_N\}$ and an objective function $f:X\rightarrow \mathbb{R}$, the goal is to find an element $x^*\in X$ such that $f(x^*)\geq f(x)$ for all $x\in X$. Let $y_i:=f(x_i)$ and define the set $Y=\{y_1,y_2,\ldots,y_N\}$ as well as the $N$-dimensional vector $
y=(y_1,y_2,\ldots,y_N)$. We focus on the case where $y_i\geq 0$ for all $i=1,\ldots,N$. This can be guaranteed for example by finding a lower bound $y_0$ such that $f(x)\geq y_0$ for all $x\in X$ and introducing a new objective function $f'(x)=f(x)-y_0,$ without altering the solution to the optimization problem.

Suppose that we can draw samples from $X$ according to some probability distribution $P(x)$, so that we now think of $Y$ as a random variable with corresponding probability distribution $P(Y = y)$. The simplest stochastic algorithm for giving an approximate solution to the above combinatorial optimization problem, known as random search, is to sample multiple values of $X$ and choose the one with the largest value of the objective function. For a uniform distribution of samples $P_{U}(x) = 1/N$, the expected value of $Y$ is
\begin{align}
\mathbb{E}_{U}(Y)=\frac{1}{N}\sum_{i=1}^Ny_i=\frac{1}{N}\Vert y\Vert _1,
\end{align}
where $\Vert \cdot \Vert_p$ is the $p$-norm. Now suppose that we instead draw samples with respect to a distribution where the probability of drawing a given element $x \in X$ is proportional to the objective function $f(x)$, i.e.,
\beq
P_{\propto}(x_i)=\frac{1}{\Vert y\Vert _1}y_i.
\eeq
We refer to this type of sampling as \emph{proportional sampling}. Here, the expected value of $Y$ is
\beq
\mathbb{E}_{\propto}(Y)=\frac{1}{\Vert y\Vert _1}\sum_{i=1}^N y_i^2=\frac{(\Vert y\Vert _2)^2}{\Vert y\Vert _1}.
\eeq
For $p$-norms on finite-dimensional vector spaces it holds that
\beq\label{Eq: p,q}
\Vert y\Vert_p\leq N^{\frac1p-\frac1q}\Vert y\Vert_q
\eeq
for any $q>p\geq 1$. The ratio between the expectation values of $Y$ with respect to the proportional and uniform distributions hence satisfies
\begin{align}
\frac{\mathbb{E}_{\propto}(Y)}{\mathbb{E}_{U}(Y)}=\left(\frac{\sqrt{N} \Vert y\Vert _2}{\Vert y\Vert _1}\right)^2\geq 1,
\end{align}
where we have used Eq. \eqref{Eq: p,q} for the case $p=1$, $q=2$. This inequality means that sampling from the proportional distribution is never worse on average than sampling from the uniform distribution. The advantage obtained from proportional sampling increases when the elements of $Y$ have very different values, in particular when there are a few elements that are much larger than the rest. Indeed, on one extreme, there is only a single non-zero element of $Y$ and the proportional distribution outputs the optimal element $x^{*}$ with certainty, thus solving the combinatorial optimization problem with a single sample. On the other extreme, all of the elements of $Y$ are equal and the proportional distribution reduces to the uniform distribution.

\subsection{Enhancing random search algorithms}

Let us build on the intuition that proportional sampling is beneficial by providing a framework for assessing the performance of random search. Suppose that we sort $X$ so that $Y$ is in non-decreasing order, i.e., $y_1\leq y_2\leq \ldots \leq y_N$. The combinatorial optimization problem depends only on the relative values of $Y$ and we can hence change this random variable so that a given $x_i$ has a corresponding $y_i$ now given by $y_i= i / N$. The value $y_i$ thus denotes the fraction of samples $x \in X$ whose objective function value $f(x)$ does not exceed $f(x_{i})$. In typical optimization problems of interest, the size of the sample space $N$ is exponentially large, so we can approximate $Y$ as a continuous random variable with values in the interval $(0,1]$ and probability density function $p(y)$. The goal of the combinatorial optimization is to find the $x^{*} \in X$ such that $Y = 1$, and we can find approximate solutions through random search by taking $\kappa$ samples of $Y$ and finding the maximum. The result of random search is another random variable
\beq
Z=\max\{Y_1,Y_2,\ldots , Y_\kappa \},
\eeq
with $Z \in (0,1]$ and its corresponding probability density function is determined by $p(y)$.

When $Y$ has a uniform distribution, i.e., $p(y)=1$, the cumulative distribution of the random variable $Z$ is
\begin{align}
P_U(Z\leq z)&=\prod_{j=1}^\kappa p(y \leq z)=z^\kappa,
\end{align}
and its probability distribution is
\beq
P_U(Z=z)=\kappa z^{\kappa-1},
\eeq
which is a Beta distribution with mean $\kappa/(\kappa+1)$. We thus conclude that the expected output of the random search algorithm using uniform sampling is
\beq
\expec_U (Z)=\frac{\kappa}{\kappa+1}.
\eeq
This should be interpreted as stating that, after taking $\kappa$ samples of $X$ from the uniform distribution, on average we will find a value of $f(x)$ that is larger than a fraction $\kappa/(\kappa+1)$ of all possible values of $f(x)$ among $x \in X$.

Instead, if $p(y)$ corresponds to a proportional distribution on $X$, then
\begin{equation}
P_{\propto}(Z\leq z) =\left[\int_0^z p(y)dy\right]^\kappa,
\end{equation}
and so
\begin{align*}
P_{\propto}(Z = z)& = \frac{d}{dz}P_{\propto}(Z\leq z)\\
&=\kappa\left[\int_0^z p(y)dy\right]^{\kappa-1}p(z).
\end{align*}
The expectation value of $Y$ when sampled from the proportional distribution is thus
\beq \label{Eq:Exp}
\expec_{\propto}(Z)=\kappa\int_0^1 z\left[\int_0^z p(y)dy\right]^{\kappa-1}p(z)dz.
\eeq
In general, the resultant expectation value will depend on the particular form of $p(y)$. To illustrate the effect of proportional sampling, we focus on the case where $p(y)$ is exponentially increasing in $y$, given by the form
\beq\label{Eq: p(x)}
p(y) =  \frac{\lambda e^{\lambda(y-1)}}{1-e^{- \lambda}},
\eeq
for some positive constant $\lambda$.
For small values of $\lambda$, the distribution is close to uniform, whereas for larger values, the probability is increasingly concentrated on the largest values, as illustrated in Fig.~\ref{Fig: Exp_Dbns}. With this distribution, the expectation value in Eq.~(\ref{Eq:Exp}) can be evaluated using standard integral tables as
\begin{align}
\expec_{\propto}(Z)=& \frac{1}{\lambda}(1-e^{\lambda})^{-\kappa} \left[ \left((1-e^{\lambda})^{\kappa} - 1\right) \lambda - H_\kappa \right. \\
& \left. - \kappa e^{\lambda}\cdot {}_3F_2(1,1,1-\kappa;2,2;e^\lambda) \right]
\end{align}
where $H_\kappa$ is the $\kappa$-th harmonic number defined as
\beq
H_\kappa=\sum_{n=1}^\kappa\frac{1}{n}
\eeq
and $_3F_2(a_1,a_2,a_3;b_1,b_2;z)$ is a generalized hypergeometric function. The improvement obtained by performing proportional sampling is illustrated in Fig. \ref{Fig: Max_Prop_sampling}. Here we plot $\alpha := -\log_{10}[1-\expec_{\propto}(Z)]$ as a function of $\kappa$ for different values of $\lambda$. This function quantifies how close the expected result of random search is to finding the optimal $X$ with value $Y=1$, i.e., for a given value of $\alpha$ then $\expec_{\propto}(Z)$ is larger than a fraction $1-10^{-\alpha}$ of all possible values of $Y$. As expected, proportional sampling always leads to a significant improvement compared to uniform sampling, with the output of the optimization algorithm $\expec_{\propto}(Z)$ being located in ranges that are orders of magnitude higher. The advantage increases with the parameter $\lambda$, showcasing the added benefits of proportional sampling when there is a small subset of elements with much larger values of the objective function. 

The benefits of sampling from probability distributions that are tailored for specific problems has long been understood. Indeed, the goal of Markov chain Monte Carlo (MCMC) algorithms is to use a source of uniform randomness to generate samples from  more general distributions. In this sense, proportional sampling can be interpreted as a type of importance sampling, a technique widely employed in optimization and Monte Carlo methods \cite{rubinstein2016simulation}, where here the importance is achieved through proportionality. Nevertheless, the overhead that arises from MCMC algorithms can be extremely costly, and there exist distributions that cannot be sampled efficiently using classical algorithms. For this reason, it is highly desirable to have access to physical devices that can generate high-rate samples from useful probability distributions. As we show in the following section, Gaussian boson sampling systems can serve that role by sampling from proportional distributions that are useful for optimization problems related to finding optimal submatrices.

Before proceeding, it is important to remark on the methods used to solve combinatorial optimization problems and the resultant usefulness of stochastic algorithms. Approximate optimization becomes relevant for any problem that takes superpolynomial time to solve. In these cases, a desirable approach is to identify a polynomial time approximation scheme. However, such a scheme does not always exist for every problem, or can be hard to identify. One alternative approach is to make use of deterministic or stochastic based heuristic algorithms to provide approximate solutions. The choice of deterministic or stochastic algorithm is problem dependent. Proportional sampling is a tool to enhance stochastic heuristic algorithms.

\begin{figure}[t!]
\includegraphics[width=0.95\columnwidth]{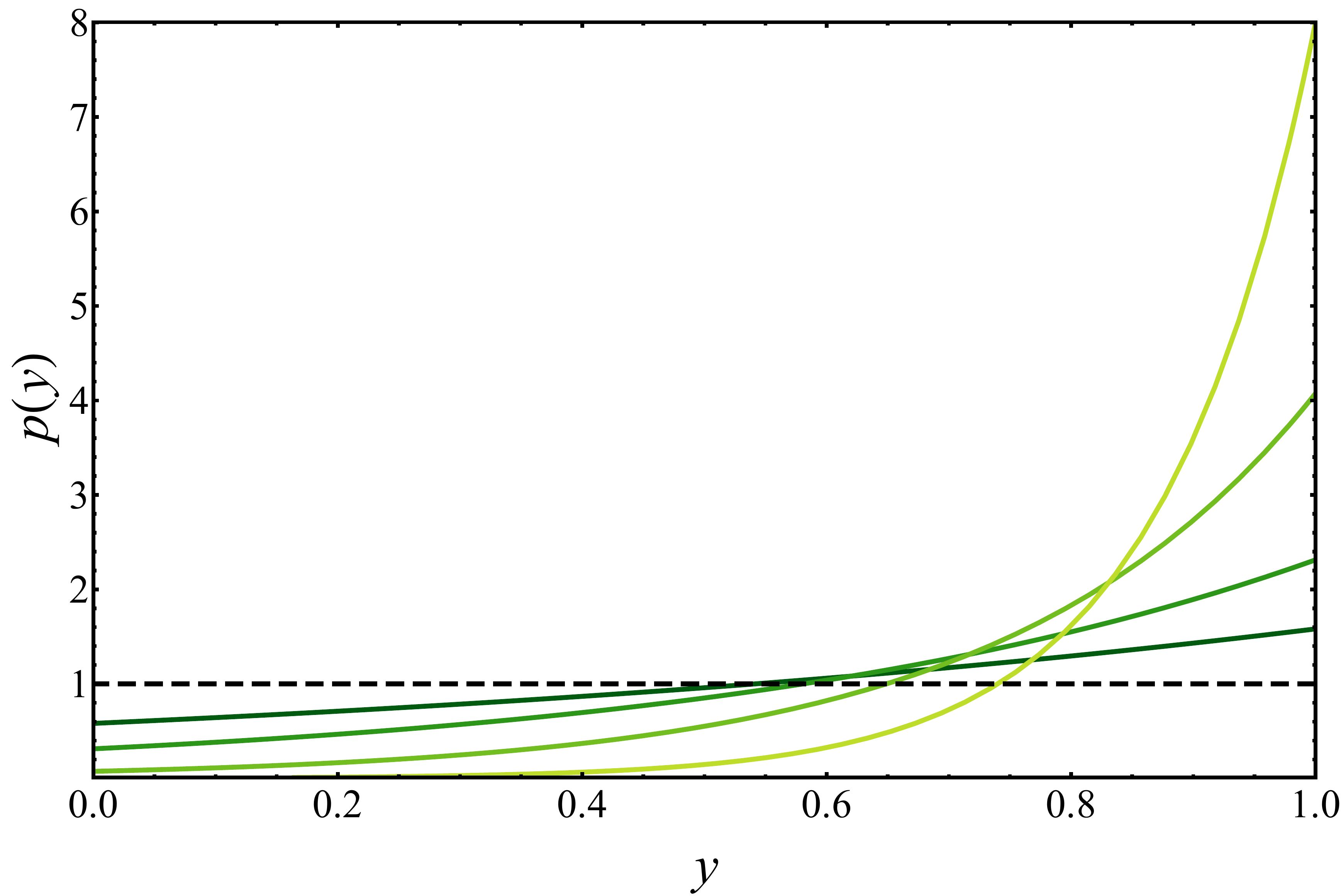}
\caption{Plot of a uniform distribution (dashed black line) and the probability distribution $p(y)=\lambda e^{\lambda(y-1)}/(1-e^{-\lambda})$ for $\lambda=1,2,4$, and $8$, where the distributions are further from uniform for increasing values of $\lambda$. }\label{Fig: Exp_Dbns}
\end{figure}

\begin{figure}[t!]
\includegraphics[width= \columnwidth]{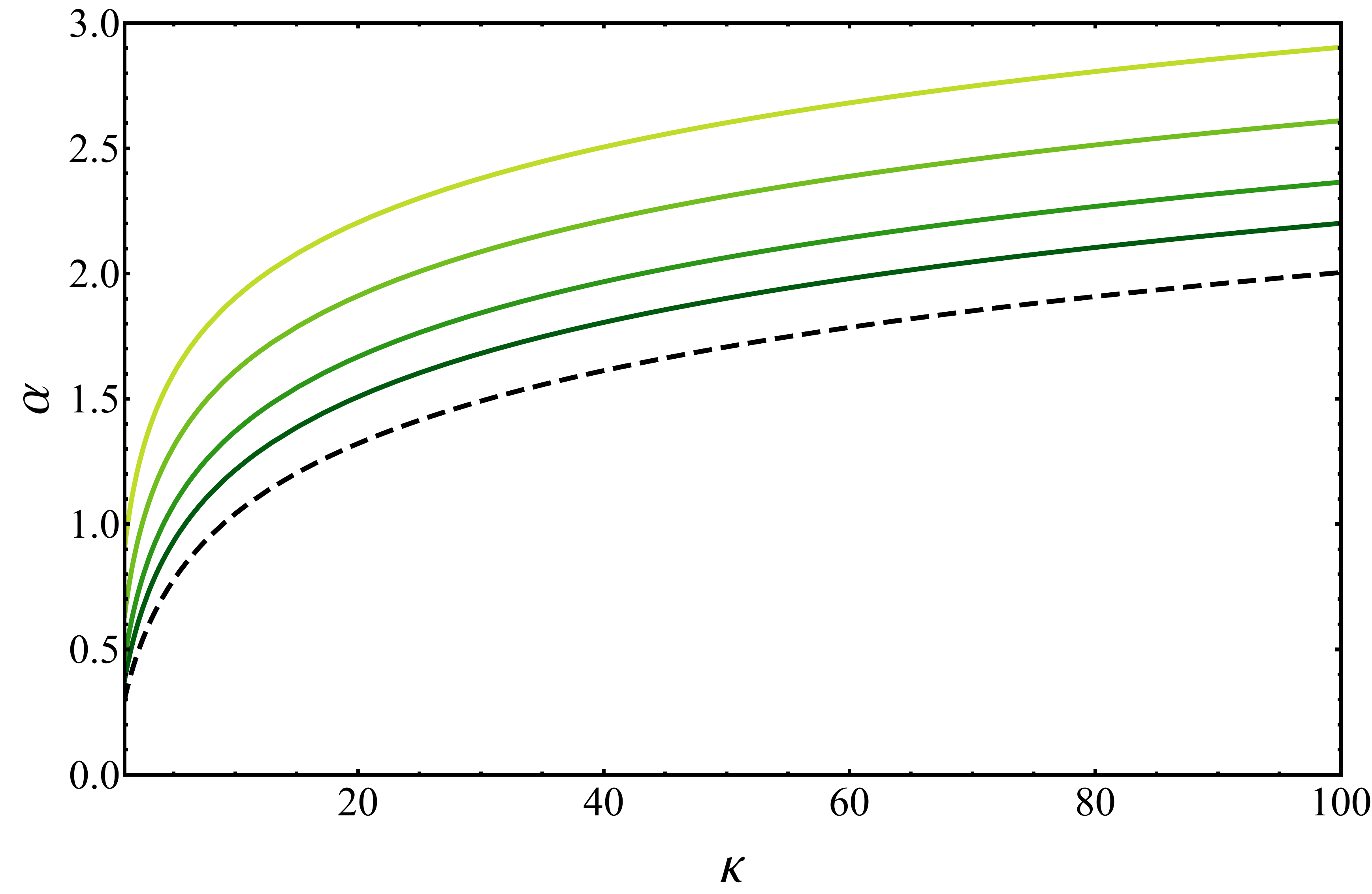}
\caption{Plot of the coefficient $\alpha$ as a function of the number of samples $\kappa$ for proportional sampling compared to uniform sampling. The dashed black line corresponds to uniform sampling and higher curves correspond to proportional sampling with distribution $p(y)=\lambda e^{\lambda(y-1)}/(1-e^{-\lambda})$ and values $\lambda=1,2,4$, and $8$. Proportional sampling always outperforms uniform sampling and the advantage increases with $\lambda$. }\label{Fig: Max_Prop_sampling}
\end{figure}

\section{Gaussian boson sampling}

Gaussian boson sampling (GBS) \cite{hamilton2017gaussian,kruse2018detailed} is a variant of boson sampling where the input to a linear optical network is a multi-mode Gaussian state instead of a collection of single photons. From an experimental perspective, this platform offers several advantages compared to traditional boson sampling and is therefore a leading candidate for the practical development of larger scale boson sampling devices. In its simplest form, the input Gaussian state in GBS is a tensor product of squeezed vacuum states in each mode. We denote the possible outputs of GBS by vectors $S=(s_1,s_2,\ldots,s_{n})$, where there are a total of $n$ input and output modes and $s_i$ is the number of photons detected in output mode $i$. It was shown in Ref. \cite{hamilton2017gaussian} that for a linear optics network characterized by the $n\times n$ unitary $U$ and for input states with equal squeezing parameter $r>0$ on all modes, the probability of observing an output pattern $S$ is given by
\beq\label{Eq: GBS_dbn}
\Pr(S)=\frac{1}{(\cosh r)^n}\frac{|\text{Haf}(B_S)|^2}{s_1!s_2!\cdots s_n!},
\eeq
where $B=\tanh r\times UU^t$ and $B_S$ is the matrix obtained by selecting the rows and columns corresponding to the non-zero entries of $S$. The function $\Haf(\cdot)$ is the Hafnian of a matrix, defined as
\beq \label{Eq:HafDef}
\Haf(X)=\sum_{\sigma\in \text{PMP}_{2m}}\prod_{j=1}^mX_{\sigma(2j-1),\sigma(2j)}
\eeq
for even dimensional matrices $X$, where $\mbox{dim}(X)=2m$, and $\Haf(X)=0$ for odd dimensional $X$. Here, $\text{PMP}_{2m}$ is the set of perfect-matching permutations on $2m$ elements. The Hafnian can be shown to be \#P-Hard to approximate in a worst-case setting, a fact that has been leveraged to argue that even approximate simulation of GBS cannot be performed in polynomial time on classical computers unless the polynomial hierarchy collapses to third level \cite{hamilton2017gaussian}.

\subsection{The Max-Haf Problem}

{We now introduce the Max-Haf combinatorial optimization problem, defined on a complex valued matrix $B$ of any dimension. For Max-Haf, the problem is to find a submatrix $B_{S}$ of $B$, among those of a fixed even dimension $k=2m$, with the largest Hafnian in absolute value, i.e.,
\begin{equation}\label{Optimization}
\begin{aligned}
&\text{given } B\text{, solve:}\\
& \underset{B_S}{\mbox{argmax}}\,|\text{Haf}(B_S)|^2\\
& \text{Subject to: }\text{dim}(B_S)= k.
\end{aligned}
\end{equation}
This problem is NP-hard, as we now prove.
\begin{thm}
The Max-Haf problem is NP-Hard to solve exactly in a worst-case setting.
\end{thm}
\begin{proof}
The proof is by reduction to the maximum clique problem, which is known to be NP-Hard \cite{bomze1999maximum}. We consider the specific case where $B$ is the adjacency matrix of a graph $G$, so that the submatrices $B_S$ are equivalent to subgraphs of $G$. The largest possible Hafnian of a graph is achieved by the complete graph, i.e., a clique, and a complete graph is the only graph that achieves this maximum value. Thus, given an algorithm to solve Max-Haf, we can iterate over all values of $k$ and output the largest $k$ such that the Hafnian of the optimal subgraph is equal to the Hafnian of a clique. This solves the maximum clique problem with $\sum_{k=1}^nk=O(n^2)$ calls, i.e., a polynomial number of calls, to the algorithm for Max-Haf. This implies that Max-Haf is also NP-Hard.
\end{proof}
The NP-hardness of the Max-Haf problem means that in general we cannot expect to find exact solutions in polynomial time and must instead settle for approximate solutions. It is important to emphasize that the Max-Haf problem is introduced here for the first time due to its fundamental link to Gaussian boson sampling, as now discussed.

Suppose that $B$ can be written as $\tanh r \times U U^{t}$. By construction, from Eq.~\eqref{Eq: GBS_dbn} we know that the GBS distribution is a proportional distribution for Max-Haf, provided that the GBS output is postselected on samples corresponding to submatrices of dimension $k$, see Eq.~\eqref{Gamma_S} below. This enables all the advantages outlined in the previous section for stochastic algorithms aimed at providing approximate solutions. Moreover, approximate sampling from the GBS distribution is infeasible for classical devices, further motivating the importance of a physical GBS device. The NP-Hardness of Max-Haf also implies that we can in principle reduce any problem in NP to Max-Haf and then employ GBS to improve the performance of stochastic optimization algorithms for the resulting problem. Alternatively, we can identify quantitative connections between the Hafnian and relevant objective functions of other optimization problems, rendering samples from a GBS device directly useful for finding approximate solutions to these problems as well. In our companion paper~\cite{arrazola2018quantum2_published}, we show how GBS can be used for finding approximate solutions to the densest $k$-subgraph problem \cite{feige2001dense}, which has several appliations in data mining \cite{kumar1999trawling,angel2012dense,beutel2013copycatch,chen2012dense}, bioinformatics \cite{fratkin2006motifcut,saha2010dense}, and finance \cite{arora2011computational}.

\subsection{Enhancement of approximate solutions through GBS}

Here we focus on instances of GBS where the interferometer unitary $U$ is drawn from the Haar measure. As shown in Ref. \cite{hamilton2017gaussian}, in this case the $n\times n$ matrix $B=\tanh r\times UU^t$ is proportional to a unitary drawn from the circular orthogonal ensemble (COE). Moreover, for $m$ sufficiently smaller than $n$, namely when $m=O(\sqrt{n})$, the matrix $B_S$ is well approximated by a symmetric random Gaussian matrix with entries drawn from the complex normal distribution $\mathcal{N}(0,\frac{\tanh r}{\sqrt{n}})$. We focus on this setting in the following to evaluate the single-shot enhancement of using GBS proportional sampling, following methods introduced in~\cite{aaronson2011computational}.

As we have discussed in the previous section, a simple classical random search algorithm for solving the corresponding Max-Haf problem proceeds by sampling even $k$-dimensional submatrices $B_S$ of $B$ uniformly at random. The algorithm generates many such samples, calculates the Hafnians of the corresponding submatrices and outputs the sample with the largest Hafnian. We represent the set of dimension $k$ submatrices of $B$ by the labels
\beq\label{Gamma_S}
\Gamma_{S,k}:=\left\lbrace S : s_i\in\{0,1\}\, \forall i, \,\sum_is_i= k \right\rbrace
\eeq
such that each element $S\in\Gamma_{S,k}$ describes the rows and columns of $B$ selected to form the submatrix $B_{S}$. Since $|\Gamma_{S,k}|=\binom{n}{k}$, the expected value of $|\Haf(B_S)|^2$ under uniform sampling satisfies
\begin{equation}\label{mu_U}
\hat{\mu}_U:=\binom{n}{k}^{-1}\sum_{S\in \Gamma_{S,k}} |\Haf(B_S)|^2.
\end{equation} 
The submatrices of $B$ are approximated by symmetric Gaussian random matrices \cite{hamilton2017gaussian}, so the above quantity is an approximation of the expectation value over the distribution $\mathcal{G}$ of these matrices, i.e., $\hat{\mu}_U \approx \mathbb{E}_{\mathcal{G}}(|\Haf(B_S)|^2)$.

We now compute the expectation value of $|\Haf(B_S)|^2$ when employing GBS, which we evaluate in terms of $k$, $n$ and $r$. Let $p$ be the probability that a GBS sample satisfies $S\in\Gamma_{S,k}$. This probability approaches unity for $k \ll n$. Furthermore, let $q_{n,r}(k)$ be the user-controlled probability of sampling the required $k$ photons when there are $n$ input squeezed modes with squeezing parameter $r$, given by
\beq
q_{n,r}(k)=\binom{\frac{n+k}{2}-1}{\frac{k}{2}}(\sech r)^n(\tanh r)^{k}.
\eeq
Using the fact that $s_1!s_2!\cdots s_n!=1$ for all $S\in\Gamma_{S,k}$, we have from Eq.~\eqref{Eq: GBS_dbn} that the expectation value of $|\Haf(B_S)|^2$ using GBS satisfies
\begin{align}
&\hat{\mu}_{GBS}:=\frac{1}{pq_{n,r}(k)}\sum_{S\in\Gamma_{S,k}}|\Haf(B_S)|^2\frac{1}{(\cosh r)^n}|\Haf(B_S)|^2\nonumber\\
&=\frac{1}{pq_{n,r}(k)(\cosh r)^n}\binom{n}{k}\binom{n}{k}^{-1}\sum_{S\in\Gamma_{S,k}}|\text{Haf}(B_S)|^4\nonumber\\
&\approx \frac{\binom{n}{k}}{q_{n,r}(k)(\cosh r)^n}\mathbb{E}_{\mathcal{G}}(|\text{Haf}(B_S)|^4)\label{mu_GBS}.
\end{align}  

Thus, to compute the expected value of $|\Haf(B_S)|^2$ when sampling from the uniform and GBS distributions, we need to compute the first two moments $\mathbb{E}_{\mathcal{G}}(|\text{Haf}(B_S)|^2)$ and $\mathbb{E}_{\mathcal{G}}(|\text{Haf}(B_S)|^4)$ of the random variable $|\text{Haf}(B_S)|^2$ with $B_{S}$ sampled from $k$-dimensional symmetric Gaussian random matrices with entries drawn from $\mathcal{N}(0,\frac{\tanh r}{\sqrt{n}})$. For simplicity, we introduce the random variable 
\beq\label{Eq: X to B_S}
X=\frac{\sqrt{n}}{\tanh r}B_S ,
\eeq
so that $X$ is now a random Gaussian matrix with entries drawn from the standard complex normal distribution $\mathcal{N}(0,1)$. It is straightforward to compute the first moment $\mathbb{E}_{\mathcal{G}}(|\text{Haf}(X)|^2)$ using the definition in Eq.~\eqref{Eq:HafDef}:
\begin{align}\label{haf^2}
&\mathbb{E}_{\mathcal{G}}\left(\sum_{\sigma,\tau\in \text{PMP}_{k}}\prod_{j=1}^n X_{\sigma(2j-1),\sigma(2j)}X^*_{\tau(2j-1),\tau(2j)}\right)\nonumber\\
&=\mathbb{E}_{\mathcal{G}}\left(\sum_{\sigma\in \text{PMP}_{k}}\prod_{j=1}^n |X_{\sigma(2j-1),\sigma(2j)}|^2\right)\nonumber\\
&=\sum_{\sigma\in \text{PMP}_{k}}\prod_{j=1}^n \mathbb{E}_{\mathcal{G}}\left(|X_{\sigma(2j-1),\sigma(2j)}|^2\right)\nonumber\\
&=\sum_{\sigma\in \text{PMP}_{k}}1\nonumber\\
&=(k-1)!!
\end{align}
The calculation for the second moment $\mathbb{E}_{\mathcal{G}}(|\text{Haf}(X)|^4)$ is significantly more involved, the details of which we defer to the Appendix. The result is
\beq\label{haf^4}
\mathbb{E}_{\mathcal{G}}(|\text{Haf}(X)|^4)=k!
\eeq
From Eqs. \eqref{Eq: X to B_S}, \eqref{haf^2} and \eqref{haf^4} we get
\begin{align}
\mathbb{E}_{\mathcal{G}}(|\text{Haf}(B_S)|^2)&=\frac{(\tanh r)^{k}}{n^{k/2}}(k-1)!!\label{haf(B)^2}\\
\mathbb{E}_{\mathcal{G}}(|\text{Haf}(B_S)|^4)&=\frac{(\tanh r)^{2k}}{n^{k}}k!\label{haf(B)^4}
\end{align}
Finally, using these values we obtain the ratio between the expectation values $\mu_{GBS}$ and $\mu_U$, given by
\beq\label{Eq: Ratio}
R:=\frac{\hat{\mu}_{GBS}}{\hat{\mu}_U}=\frac{\binom{n}{k}k!!}{\binom{\frac{n+k}{2}-1}{\frac{k}{2}}n^{k/2}}.
\eeq
This ratio characterizes the advantage obtained by using GBS, as illustrated in Fig \ref{Fig: Hafnian advantage}. Note that this ratio does not depend on the squeezing parameter $r$, whose only role is to determine the probability $q_{n,r}(k)$ of observing the desired number of photons.

\begin{center}
\begin{figure}[t!]
\includegraphics[width= 0.9\columnwidth]{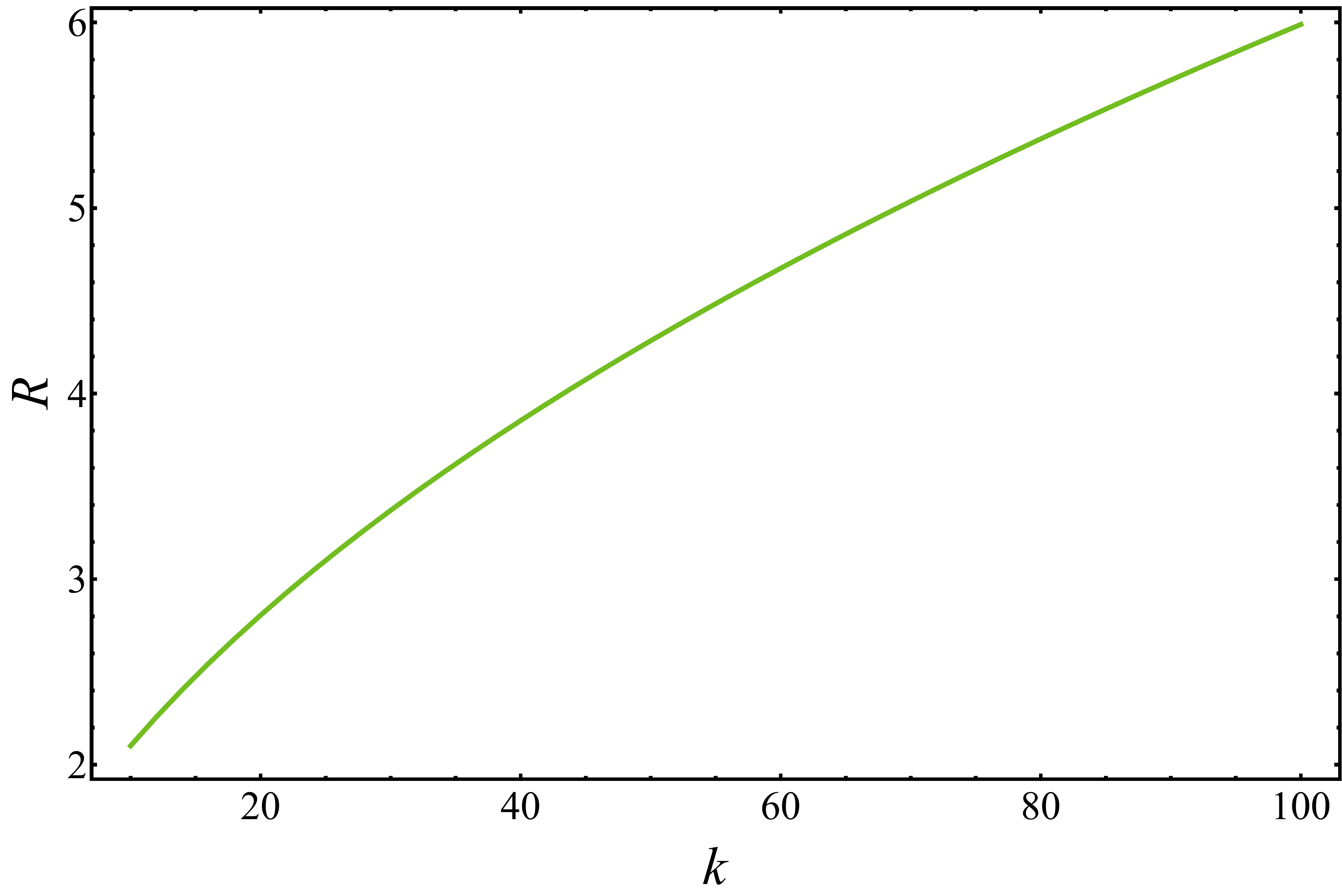}
\caption{Plot of the ratio $R$ of Eq. \eqref{Eq: Ratio} quantifying the improvement of using GBS sampling as opposed to uniform sampling for finding approximate solutions to the Max-Haf problem with submatrix size $k$. We plot $R$ for different values of even $k$ when $n=k^2$. The ratio is larger than 1 for all $k\geq 2$, showcasing the advantage obtained by using GBS.}\label{Fig: Hafnian advantage}
\end{figure}
\end{center}

This result quantifies the fundamental feature that makes GBS useful: whenever we sample a submatrix, the GBS distribution has a preference for selecting submatrices that have a large absolute Hafnian, while neglecting those with negligible Hafnians. The net result of this effect is that submatrices with larger Hafnians are on average sampled, leading to better approximations to the Max-Haf problem in a random search algorithm. However, random search algorithms, although powerful, do not exploit the local structure of the optimization landscape. Alternative stochastic algorithms employ a combination of exploration of the landscape -- which is done by random sampling -- and exploitation of local structure. We now outline how GBS can be used to enhance both of these elements within stochastic algorithms.

\section{Using GBS to enhance stochastic optimization algorithms}
Stochastic algorithms usually employ randomness in two ways: exploration and tweaking, which is used in the exploitation phase. Let us first show how GBS can be used to enhance exploration, which in the case of Max-Haf corresponds to a random selection of a set of submatrices $B_S$. Optimization algorithms usually perform this step either by selecting submatrices uniformly at random, or by seeding them, i.e., picking a few which are believed to be good submatrices. This latter strategy is often infeasible when the number of possibilities is exponentially large and may also restrict the resultant approximate solutions to sub-optimal local maxima since there are no guarantees that the initial picks are indeed good candidates. For this reason, uniform initialization is used as a reliable approach. However, as we have shown, it is always preferable to replace uniform sampling with proportional sampling. In the case of Max-Haf, this corresponds to sampling from the suitably postselected GBS distribution. We thus define a function, $\textsf{GBS-Explore}$, to be employed in heuristic algorithms:\\
\begin{defn}
$\textsf{GBS-Explore}(k,B)$: Given the matrix $B$, generate a sample $S$ from the conditional GBS distribution $\Pr(S|S \in \Gamma_{S,k})=\frac{1}{pq_{n,r}(k)(\cosh r)^n}|\text{\emph{Haf}}(B_S)|^2$. Output the corresponding matrix $B_S$.\\
\end{defn}
In a physical implementation, sampling from $\Pr(S|S \in \Gamma_{S,k})$ can be performed by generating outcomes from the complete distribution $\Pr(S)$ and keeping only outcomes with $k$ photons. This incurs a penalty of $pq_{n,r}(k)$ in terms of the sampling rate $q_{n,r}(k)$, which can be optimized by selecting an appropriate value of the squeezing parameter $r$, and the valid sample probability $p$, which approaches $1$ for $k \ll n$.

Tweaking is an operation used within exploitation where samples are randomly modified to produce new outputs in the vicinity of the original sample. In our case, since each possible submatrix is specified by a binary vector $S=(s_1,s_2,\ldots,s_n)$ with $\sum_{i=1}^n=k$, the only freedom for modification is to change the location of some of the non-zero entries of $S$. This is achieved by randomly selecting a subset of non-zero entries and setting them to zero, and then randomly selecting a subset of zero entries and setting them to one. To make use of GBS in this tweaking stage, we sample from the GBS distribution to identify which entries to change from zero to one. Here, the idea is that GBS will prefer to highlight new submatrices with a large absolute Hafnian to substitute into the candidate submatrix $B_{S}$. We formalize this into the subroutine $\textsf{GBS-Tweak}$, where $\ell \in  \{0,2, 4, \ldots, k - 2\}$ fixes the minimum number of non-zero entries of $S$ to be left unchanged:\\
\begin{defn}
\emph{$\textsf{GBS-Tweak}(S,\ell,B)$}: \begin{enumerate}
\item Generate a binary vector $S'$ with $\ell + L$ entries randomly selected from the non-zero entries of $S$ and with the remaining entries set to zero, where $L\in \{0,1,\ldots, k - \ell - 1\}$ is selected uniformly at random.
\item Generate a sample $T$ from the conditional GBS distribution $\Pr(T|\Gamma_{T,\ell})=\frac{1}{pq_{n,r}(\ell)(\cosh r)^{2n}}|\text{\emph{Haf}}(B_{T})|^2$, producing samples $T$ with $\ell$ non-zero entries.
\item Fix a binary vector $T'$ by selecting $L$ of the non-zero entries of $T$ at random and setting the remaining entries to zero. If there is any overlap between the non-zero entries of $T'$ and $S'$, repeat this step.
\item Calculate the vector $R$ given by combining the non-zero entries of $S'$ with $T'$ and return the submatrix $B_{R}$. 
\end{enumerate}
\end{defn}
The result of $\textsf{GBS-Tweak}$ is a submatrix $B_{R}$ which has been tweaked from $B_{S}$ by using GBS to preferentially select submatrices of size $\ell<k$ with large absolute Hafnians.

\textsf{GBS-Explore} and \textsf{GBS-Tweak} can be used to construct quantum-enhanced versions of any stochastic algorithm employing randomness for exploration and tweaking. This is the essential merit of GBS as a tool for approximate optimization: anything that a classical stochastic optimization algorithm can do, we can do better by enhancing it using GBS randomness. More generally, since classical algorithms must always operate on an initial source of uniform randomness, even algorithms whose goal is to sample from non-uniform distributions can be expected to improve by replacing their initial uniform sampling with GBS. 

\subsection{Applying GBS to example algorithms}

\begin{algorithm}[t!]
  \caption{GBS Simulated Annealing for Max-Haf}
\begin{algorithmic}\label{Algorithm:GBSSA}
\STATE $B$ : Input matrix
\STATE $k$ : Dimension of submatrix
\STATE $\ell$ : Minimum number of entries to change when tweaking
\STATE $t$ : Initial temperature
\STATE $a_{\text{max}}$ : Number of steps
\STATE $B_S =\textsf{GBS-Explore}(k,B)$
\STATE Best = $B_S$
\STATE \textbf{for} a from 1 to $a_{\text{max}}$:
\STATE $\hspace{5mm}$ $B_R = \textsf{GBS-Tweak}(S,\ell,B)$
\STATE $\hspace{5mm}$ \textbf{if} $|$Haf($B_R$)$|$ $>$ $|$Haf($B_S$)$|$:
\STATE $\hspace{10mm}$ $B_S = B_R$
\STATE $\hspace{5mm}$ \textbf{else}:  
\STATE $\hspace{10mm}$ Set $B_S = B_R$ with prob. $\exp[\frac{|\text{Haf}(B_R)| - |\text{Haf}(B_S)|}{t}]$
\STATE $\hspace{5mm}$ \textbf{if} $|$Haf($B_S$)$|$ $>$ $|$Haf(Best)$|$:
\STATE $\hspace{10mm}$ Best = $B_S$
\STATE $\hspace{5mm}$ Decrease $t$
\STATE \textbf{end for}
\STATE Output Best, $\Haf(\text{Best})$
\end{algorithmic}
\end{algorithm}

To exhibit the enhancement of stochastic algorithms through GBS for approximate optimization, we focus on two well-known algorithms: simulated annealing~\cite{vanLaarhoven1987simulated} and a greedy algorithm. We describe each algorithm and propose the enhanced versions using the \textsf{GBS-Explore} and \textsf{GBS-Tweak} subroutines. A comparison of performance between the two versions of each algorithm, i.e., enhanced with GBS sampling and the traditional uniform sampling, is then carried out in the following section.

Simulated annealing is a heuristic optimization algorithm that combines elements of random search and hill climbing. It begins with an exploration phase, where an initial submatrix is generated, and then proceeds to exploit the local landscape by repetitively tweaking the submatrix. For each tweak, if the absolute Hafnian of the proposal submatrix is larger than the current one, it is kept. If its absolute Hafnian is smaller, the proposal submatrix can still be retained with a probability that depends on the difference between the absolute Hafnians and a time-evolving temperature parameter. The temperature is initially set high, so new submatrices with lower absolute Hafnian are often accepted, leading effectively to a random search which avoids getting stuck in local minima. As the algorithm progresses, the temperature is lowered and only better submatrices are kept, leading to an effective hill-climbing behavior. The GBS version of simulated annealing is shown as Algorithm~\ref{Algorithm:GBSSA}.

The greedy algorithm begins with exploration by randomly generating a submatrix $B_S$ from GBS. It then proceeds to generate candidate replacement submatrices by exhaustively substituting the first row and column of $B_{S}$ with all possible remaining rows and columns of $B$, while keeping all other rows and columns of $B_S$ fixed. The best of these candidates is kept and set as $B_S$. The process is repeated for all rows and columns of $B_S$. Note that tweaking in the greedy algorithm is deterministic, and hence cannot be enhanced through GBS randomness. However, the initial exploration step can still be improved, leading us to the GBS Greedy Algorithm outlined in Algorithm \ref{Algorithm:GBSGreedy}. This algorithm can be repeated many times with a new random initial submatrix in each run. The resulting output is the best submatrix among all repetitions. The quantum enhancement of this hybrid algorithm arises because the starting submatrix will on average have a larger absolute Hafnian than one obtained from uniform sampling. 

\begin{algorithm}[t!]
  \caption{GBS Greedy Algorithm for Max-Haf}
\begin{algorithmic}\label{Algorithm:GBSGreedy}
\STATE $B$ : Input matrix
\STATE $k$ : Dimension of submatrix
\STATE $B_S=\textsf{GBS-Explore}(k,B)$
\STATE \textbf{for} $i$ from 1 to $k$:
\STATE $\hspace{5mm}$ BestHaf = 0
\STATE $\hspace{5mm}$ Bestj = 0
\STATE $\hspace{5mm}$ \textbf{for} $j$ from 1 to $n$:
\STATE $\hspace{10mm}$ $C=$ Matrix obtained by replacing $i$-th row/column \vspace{-0.35cm}
\STATE $\hspace{10mm}$ of $B_S$ with $j$-th row/column of $B$.  
\STATE $\hspace{10mm}$ \textbf{if} $|\text{Haf}(C)|>$ BestHaf:
\STATE $\hspace{15mm}$ BestHaf = $|\text{Haf}(C)|$
\STATE $\hspace{15mm}$ Bestj = j
\STATE $\hspace{10mm}$ \textbf{end if}
\STATE $\hspace{5mm}$ \textbf{end for}
\STATE $\hspace{5mm}$ Update $B_S$ by replacing its $i$-th row/column with the \STATE $\hspace{5mm}$ row/column corresponding to Bestj.
\STATE \textbf{end for}
\STATE Output $B_{S}$, $\Haf(B_S)$
\end{algorithmic}
\end{algorithm}

\subsection{Numerical study of enhancement through GBS}

Here we quantitatively analyze the performance of random search, simulated annealing and greedy algorithms in both their GBS and uniform versions. To do this, we setup an example Max-Haf problem by generating a $30\times 30$ random matrix $B=\tanh r\times UU^t$, where $U$ is drawn from the Haar measure. We fix the instance of the problem to be finding the submatrix of dimension 10 with the largest absolute Hafnian. To exactly perform GBS sampling, we carry out a brute force calculation of the GBS probability distribution conditioned on $10$ output photons with no more than one photon in each mode~\cite{bjorklund2018faster}, and then sample from the resulting distribution. The brute force calculation also allows us to find the exact solution of Max-Haf. Note that the dimension of the matrix and submatrix in this instance of Max-Haf are chosen to allow us to perform the numerical simulations in reasonable time.

We report results for algorithms making 1000 evaluations of the objective function, which is the largest overhead in each algorithm. For a random search and simulated annealing, the number of evaluations is equivalent to the number of samples. However, for the greedy algorithm, given an initial submatrix of dimension 10, the algorithm iterates over at least $21\times 10=210$ other submatrices, evaluating the Hafnian for each. Thus, to compare fairly to the results of random search and simulated annealing, we restrict ourselves to 5 repetitions of the greedy algorithm. Our results complement the study in~\cite{arrazola2018quantum2_published} testing the enhancement through GBS of random search and simulated annealing in finding dense subgraphs.

\begin{center}
\begin{figure}[h!]
\includegraphics[width= 0.96 \columnwidth]{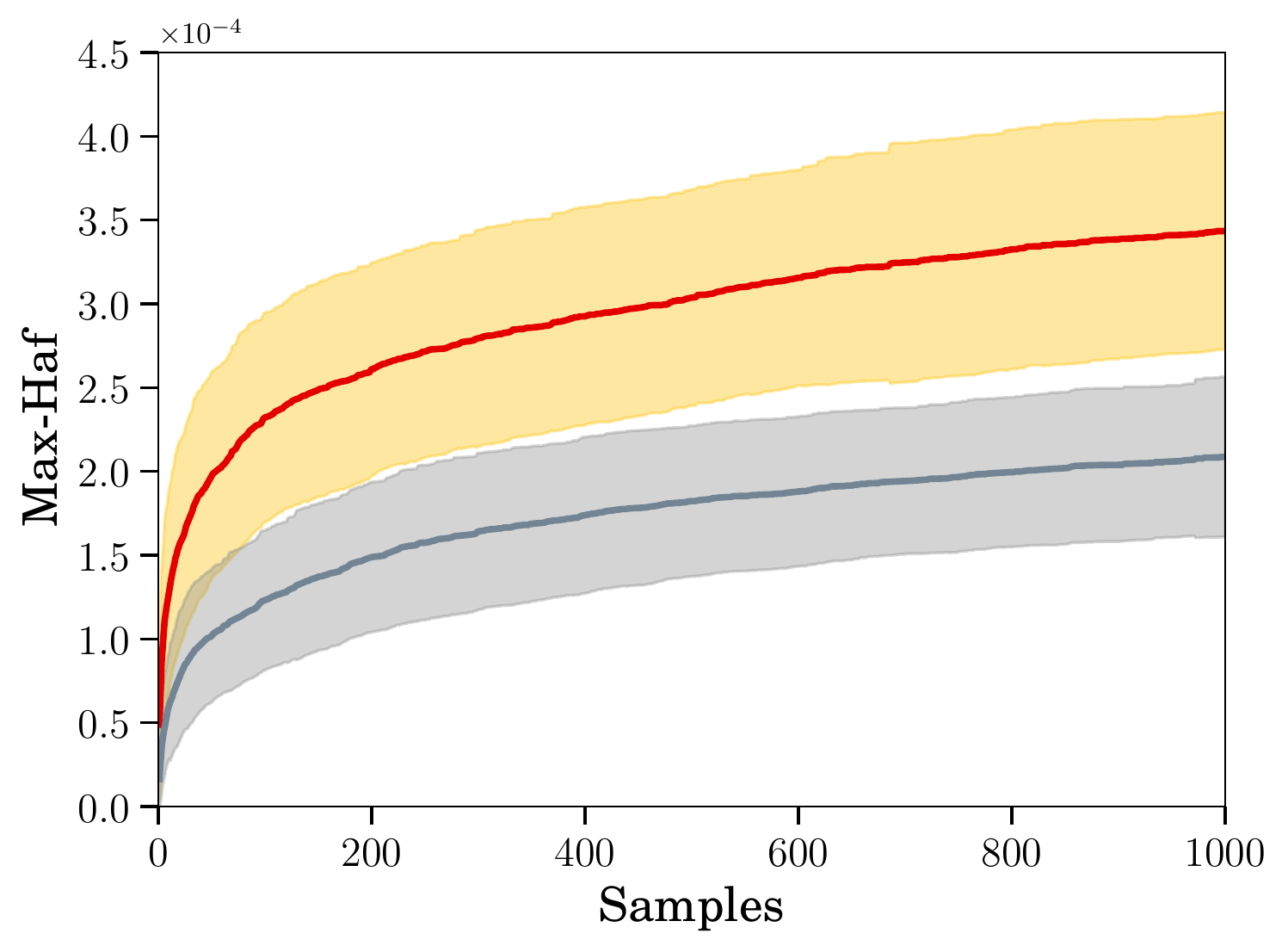}\\
\includegraphics[width= 0.96 \columnwidth]{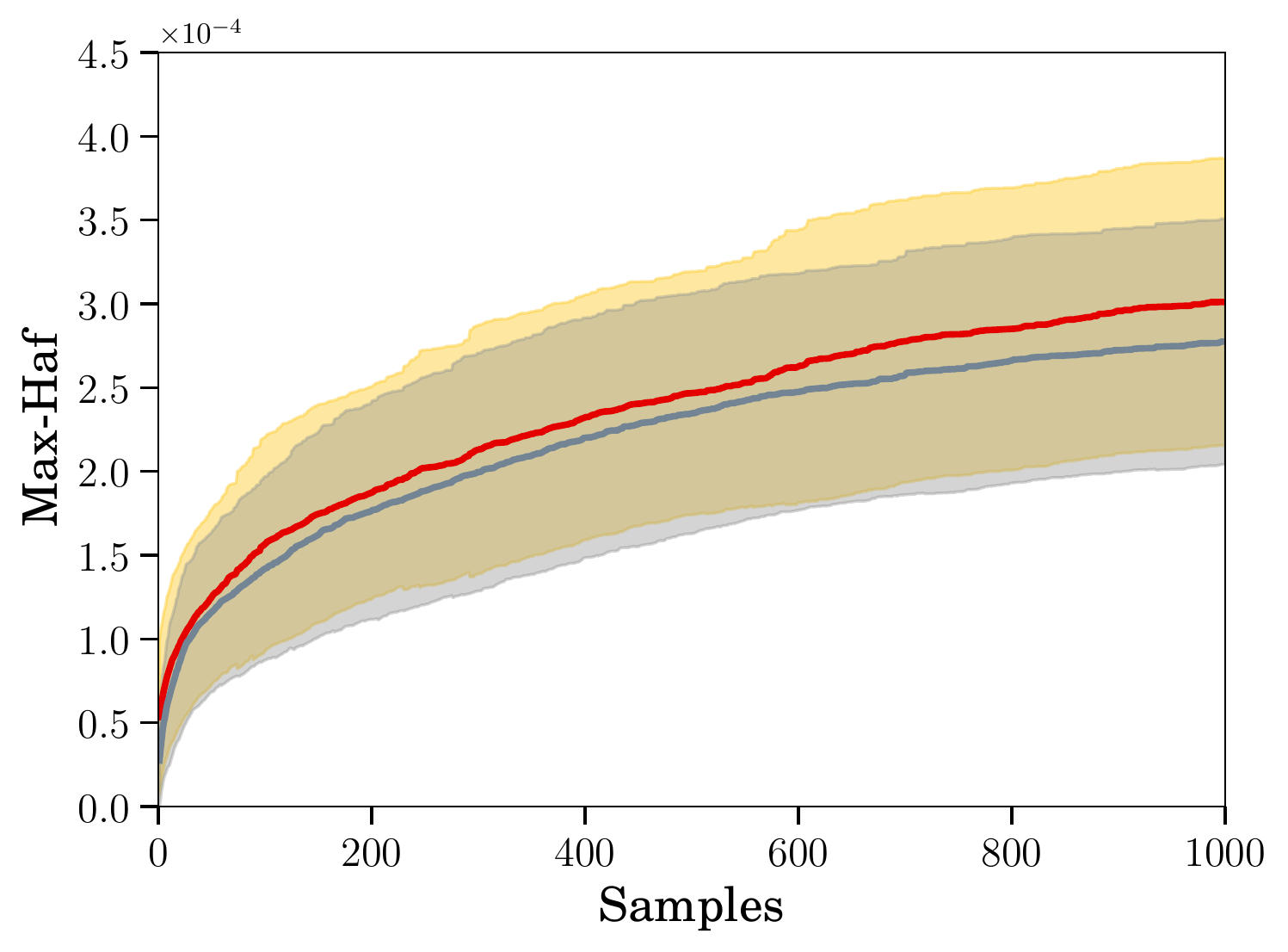}\\
\includegraphics[width= 0.96 \columnwidth]{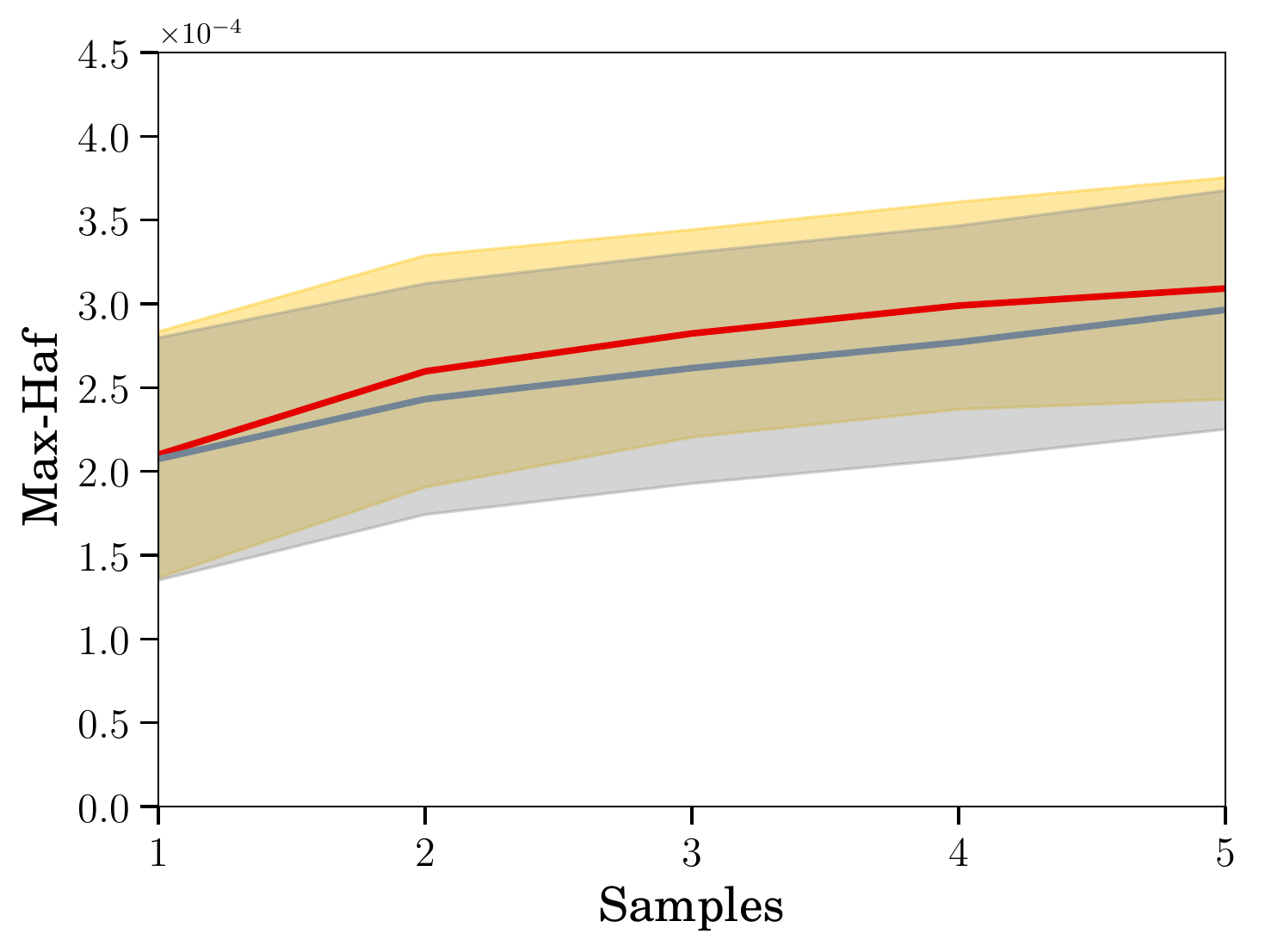}\\
\caption{Performance of random search (top), simulated annealing (middle) and  greedy (bottom) algorithms for Max-Haf. The goal is to find an optimal submatrix of dimension 10 from a $30\times 30$ starting matrix $B$ drawn from the circular orthogonal ensemble. The top red curves correspond to GBS random search and the bottom grey curves to uniform random search. The solid curve is the average observed over $400$ repetitions and the continuous error bars represent one standard deviation. Max-Haf can be solved by brute force with result $7.17 \times 10^{-4}$.} \label{Fig: Combined}
\end{figure}
\end{center}

The performances of random search, simulated annealing, and the greedy algorithm are shown in Fig. \ref{Fig: Combined}. The simulated annealing parameters were set to $T = 3 \times 10^{-5}$ and $\ell = 6$, with a linear cooling schedule used. As expected, we observe an improvement when using GBS over uniform sampling in all cases. The advantage is most significant for random search, which is the algorithm harnessing the most randomness. The improvement is more modest for simulated annealing and the greedy algorithm. This is to be expected for the greedy algorithm since it is mostly deterministic: the only random component is the initial exploration. On the other hand, the small advantage for simulated annealing indicates that GBS-tweak is only slightly more useful than a uniform tweak. For certain situations, as in Ref.~\cite{arrazola2018quantum2_published}, it is possible to further optimize GBS-tweak to increase the overall advantage. A striking feature of these results is that the GBS random search is the best performing algorithm among all variations, whereas uniform random search is the worst. This highlights the power of proportional sampling for optimization: even the simplest algorithm we have studied, when enhanced with GBS, surpasses the capabilities of more sophisticated classical algorithms. 

The algorithms we have discussed so far require several evaluations of a Hafnian, which is generally a \#P-Hard quantity to compute. For large problem sizes, when this evaluation becomes infeasible, it may be desirable to employ alternative algorithms such as Bayesian optimization \cite{mockus2012bayesian} that are more suitable for handling objective functions that are costly to evaluate. We reiterate that random search, greedy algorithms and simulated annealing are only examples, but the improvement from GBS can occur for any stochastic algorithm relying on exploration and tweaking. It is important to further investigate the extent to which GBS-enhanced stochastic algorithms for Max-Haf may outperform purely classical methods, in particular deterministic algorithms that are not straightforwardly improved by using GBS.

Finally, the Max-Haf problem is the canonical optimization task for GBS, so studying it provides the best insights for understanding how GBS can be applied to approximate optimization. However, the true potential of GBS is unlocked when we carry these insights over to a wider class of optimization problems, as shown in detail in Ref.~\cite{arrazola2018quantum2_published} for the densest $k$-subgraph problem.  

\section{Discussion}
We have argued that Gaussian boson sampling, besides providing an avenue for challenging the extended Church-Turing thesis, can be a valuable tool in enhancing stochastic algorithms for NP-Hard optimization problems. Our results provide additional evidence to challenge the common perception that boson sampling devices, while performing a task that is hard to replicate with classical computers, are not immediately relevant to problems of practical interest. Furthermore, classical sampling algorithms employ deterministic rules and sources of uniform randomness to generate samples from target distributions. This approach can often introduce significant overheads and it is incapable of efficiently sampling from arbitrary distributions. Boson sampling devices, and quantum computers more generally, when programmed in a deliberate and engineered fashion, have the potential to efficiently sample from a larger class of probability distributions, leading to improved algorithms for optimization and simulation.

Future work should focus on expanding the range of applications of Gaussian boson sampling for approximate optimization. One particularly appealing direction is to understand how we can generate proportional distributions for problems beyond submatrix optimization, for instance by employing techniques from machine learning. It will also be important for a given problem to benchmark the enhancement through GBS against optimized classical algorithms.

\textit{Acknowledgements.---} The authors thank Alex Arkhipov, Kamil Br{\' a}dler, Pierre-Luc Dallaire-Demers, Nathan Killoran, Seth Lloyd and Christian Weedbrook for valuable discussions.

\appendix
\section*{Appendix}
We perform a detailed calculation of the quantity $\mathbb{E}_{\mathcal{G}}(|\text{Haf}(X)|^4)$. Our calculation follows closely a strategy from Ref. \citep{aaronson2011computational} to compute analogous moments of permanents of Gaussian random matrices. We have
\begin{widetext}
\begin{align*}
\mathbb{E}_{\mathcal{G}}(|\text{Haf}(X)|^4)&=\mathbb{E}_{\mathcal{G}}\left(\sum_{\sigma,\tau,\alpha,\beta\in\text{PMP}_{2m}}\prod_{j=1}^n X_{\sigma(2j-1),\sigma(2j)}X_{\tau(2j-1),\tau(2j)}X^*_{\alpha(2j-1),\alpha(2j)}X^*_{\beta(2j-1),\beta(2j)}\right)\\
&=\sum_{\sigma,\tau,\alpha,\beta\in\text{PMP}_{2m}}\mathbb{E}_{\mathcal{G}}\left(\prod_{j=1}^n X_{\sigma(2j-1),\sigma(2j)}X_{\tau(2j-1),\tau(2j)}X^*_{\alpha(2j-1),\alpha(2j)}X^*_{\beta(2j-1),\beta(2j)}\right)\\
&=\sum_{\sigma,\tau,\alpha,\beta\in\text{PMP}_{2m}}M(\sigma,\tau,\alpha,\beta),
\end{align*}
where we have implicitly defined the function $M(\sigma,\tau,\alpha,\beta)$.
We write $\sigma\cup\tau=\alpha\cup\beta$ if
\begin{align*}
&\{(\sigma(1),\sigma(2)),(\tau(1),\tau(2)),\ldots,(\sigma(2m-1),\sigma(2m)),(\tau(2m-1),\tau(2m-1))\}=\\
&\{(\alpha(1),\alpha(2)),(\beta(1),\beta(2)),\ldots,(\alpha(2m-1),\alpha(2m)),(\beta(2m-1),\beta(2m-1))\}.
\end{align*}
\end{widetext}
Note that $M(\sigma,\tau,\alpha,\beta)\neq 0$ if and only if $\sigma\cup\tau=\alpha\cup\beta$, so we can restrict our attention to the case $\sigma\cup\tau=\alpha\cup\beta$. For each $j=1,2,\ldots,n$ there are two possibilities. If $\sigma(2j-1),\sigma(2j)\neq \tau(2j-1),\tau(2j)$ then
\begin{align*}
&\mathbb{E}_{\mathcal{G}}(X_{\sigma(2j-1),\sigma(2j)}X_{\tau(2j-1),\tau(2j)}X^*_{\alpha(2j-1),\alpha(2j)}X^*_{\beta(2j-1),\beta(2j)})\\
&=\mathbb{E}_{\mathcal{G}}(|X_{\sigma(2j-1),\sigma(2j)}|^2|X_{\tau(2j-1),\tau(2j)}|^2)\\
&=\mathbb{E}_{\mathcal{G}}(|X_{\sigma(2j-1),\sigma(2j)}|^2]\mathbb{E}_{\mathcal{G}}[|X_{\tau(2j-1),\tau(2j)}|^2)=1.
\end{align*}
Similarly if $\sigma(2j-1),\sigma(2j)= \tau(2j-1),\tau(2j)$ we have
\begin{align*}
&\mathbb{E}_{\mathcal{G}}(X_{\sigma(2j-1),\sigma(2j)}X_{\tau(2j-1),\tau(2j)}X^*_{\alpha(2j-1),\alpha(2j)}X^*_{\beta(2j-1),\beta(2j)})\\
&=\mathbb{E}_{\mathcal{G}}(|X_{\sigma(2j-1),\sigma(2j)}|^4)=2.
\end{align*}
We conclude that, whenever $\sigma\cup\tau=\alpha\cup\beta$, it holds that
\beq
M(\sigma,\tau,\alpha,\beta)=2^{K(\sigma,\tau)},
\eeq
where $K(\sigma,\tau)$ is the number of $j$ such that $\sigma(2j-1),\sigma(2j)= \tau(2j-1),\tau(2j)$. Finally, we obtain
\begin{align*}
\mathbb{E}_{\mathcal{G}}(|\text{Haf}(X)|^4)&=\sum_{\sigma,\tau\in \text{PMP}_{2m}}2^{K(\sigma,\tau)}N(\sigma,\tau)
\end{align*}
where $N(\sigma,\tau)$ is the number of permutations $\alpha,\beta$ such that $\sigma\cup\tau=\alpha\cup\beta$. This expression can be simplified further as follows. Let $\sigma_0$ be the identity permutation, which is a perfect matching permutation. Furthermore, for a given perfect matching permutation $\sigma$, let $\sigma^{-1}$ be the unique perfect matching permutation such that $\sigma^{-1}\sigma=\sigma_0$. We then have
\begin{align*}
\mathbb{E}_{\mathcal{G}}(|\text{Haf}(X)|^4)&=\sum_{\sigma,\tau\in \text{PMP}_{2m}}2^{K(\sigma,\tau)}N(\sigma,\tau)\\
&=[(2m-1)!!]^2\underset{\sigma,\tau}{\mathbb{E}_{\mathcal{G}}}(2^{K(\sigma,\tau)}N(\sigma,\tau))\\
&=\underset{\sigma,\tau}{\mathbb{E}_{\mathcal{G}}}(2^{K(\sigma^{-1}\sigma,\sigma^{-1}\tau)}N(\sigma^{-1}\sigma,\sigma^{-1}\tau))\\
&=\underset{\xi\in\text{PMP}_{2m}}{\mathbb{E}_{\mathcal{G}}}(2^{K(\sigma_0,\xi)}N(\sigma_0,\xi))\\\
&=(2m-1)!!\sum_{\xi\in\text{PMP}_{2m}}2^{K(\sigma_0,\xi)}N(\sigma_0,\xi).
\end{align*}
We now need to calculate $\sum_{\xi\in\text{PMP}_{2m}}2^{K(\sigma_0,\xi)}N(\sigma_0,\xi)$.\\

Define a graph $G_{\xi}(m)$ of $2m$ nodes as follows. For all $j=1,2,\ldots m$, draw an edge between nodes $(2j-1,2j)$. This is the graph representation of the identity permutation $\sigma_0$. Similarly, for $\xi=(i_1j_1)(i_2j_2)\ldots(i_{m}j_{m})$, draw an edge between nodes $(i_k,j_k)$ for all $k=1,2,\ldots,m$. This is the graph representation of $\xi$.  An example of a graph $G_{\xi}(m)$ is shown in Fig. \ref{Fig: CycleGraph}. Note that, by construction, $G_{\xi}(m)$ is a union of cycles.

\begin{center}
\begin{figure}[t!]
\includegraphics[width= 0.45\columnwidth]{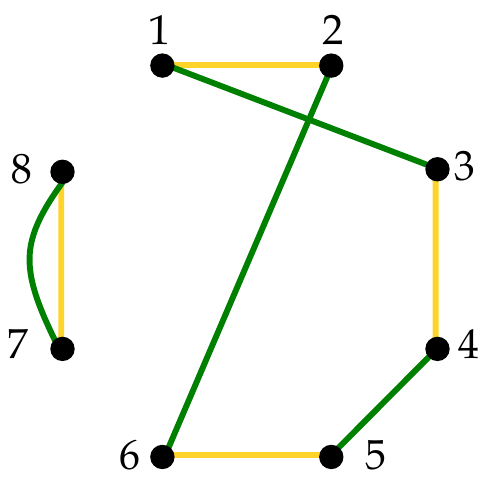}
\caption{Graph $G_{\xi}(4)$ for the perfect matching permutation $\xi=(13)(26)(45)(78)$. There is a 1-cycle corresponding to the equivalent action of the two permutations on the nodes 7 and 8. The remaining edges form a 3-cycle for which there are two inequivalent choices of permutations $\alpha,\beta$ such that $\sigma_0\cup\xi=\alpha\cup\beta$, namely: (i) $\alpha=(12)(34)(56)$, $\beta=(13)(26)(45)$ and (ii) $\alpha=(13)(26)(45)$, $\beta=(12)(34)(56)$. }\label{Fig: CycleGraph}
\end{figure}
\end{center}

\begin{center}
\begin{figure}[t!]
\includegraphics[width= \columnwidth]{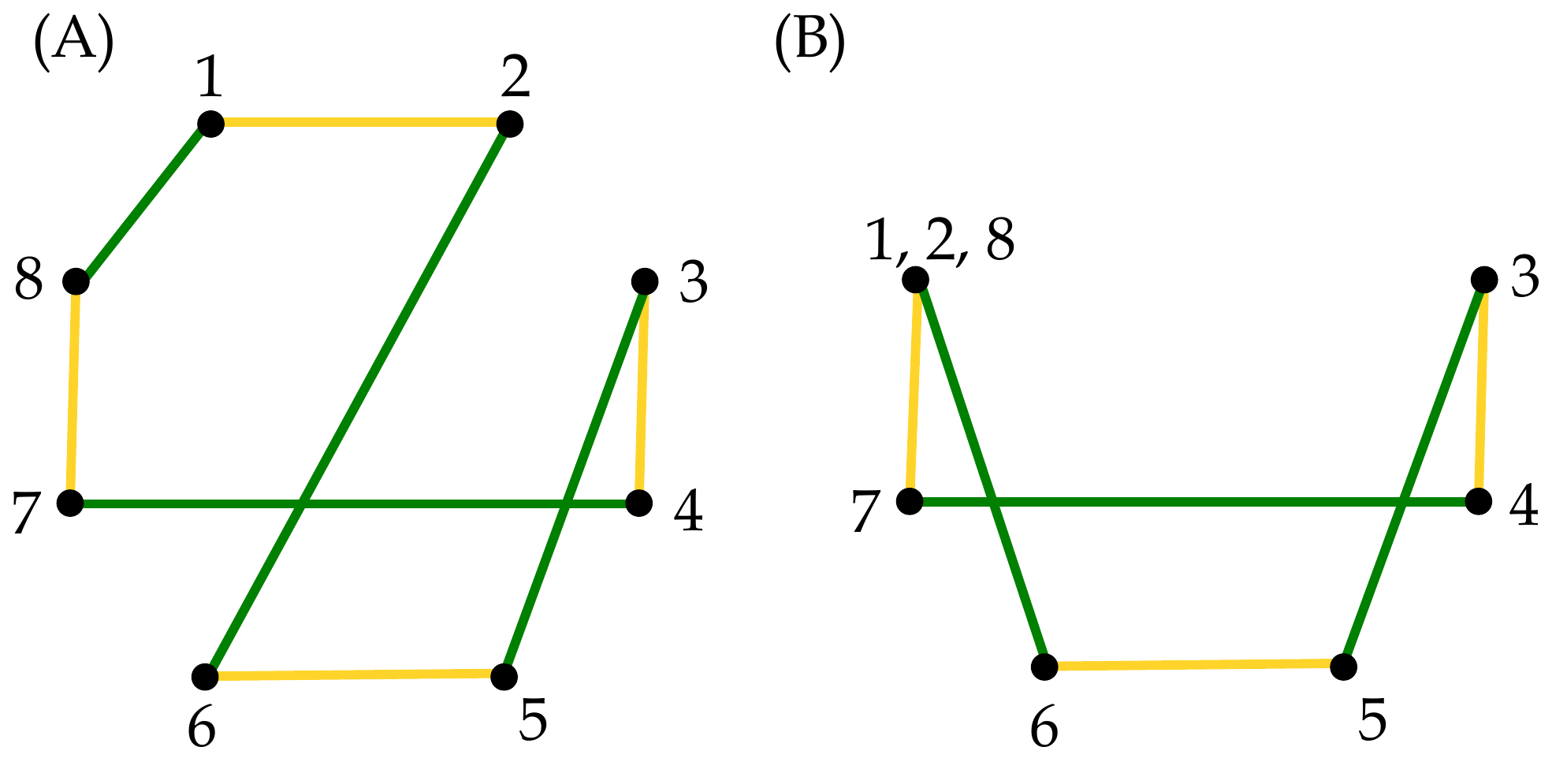}
\caption{When the permutation $\xi$ has a swap of the form $(1j)$ with $j\neq 2$, we can combine the nodes $1,2$ and $j$ into a single node without changing the number of cycles in the graph $G_{\xi}(m)$. Figure (A) shows the graph $G_{\xi}(4)$ for the perfect matching permutation $\xi=(18)(26)(35)(26)$ and figure (B) shows a graph with an equal number of cycles obtained by combining nodes 1,2, and 8 into a single node.}\label{Fig: EquivalentGraph}
\end{figure}
\end{center}

Whenever $(\sigma_0(2j-1),\sigma_0(2j))=(\xi(2j-1),\xi(2j))$, there will be a 1-cycle in the graph $G_{\xi}(m)$ consisting of the two edges joining nodes $2j-1$ and $2j$. Thus, $K(\sigma_0,\xi)$ is equal to the number of 1-cycles in $G_{\xi}(m)$. Every larger cycle in $G_{\xi}(m)$ gives rise to 2 inequivalent choices of 
$\alpha,\beta$ such that $\sigma_0\cup\xi=\alpha\cup\beta$ because we can assign the action of $\alpha$ to coincide with either $\sigma_0$ or $\xi$ when restricted to the nodes in the cycles, and similarly for $\beta$. We therefore have that 
\beq
N(\sigma_0,\xi)=2^{\Gamma_{\xi}}
\eeq
where $2^{\Gamma_{\xi}}$ is the number of cycles in $G_{\xi}(m)$ of length equal or larger than 2. Combining these results we conclude that
\begin{align}
\mathbb{E}_{\mathcal{G}}(|\text{Haf}(X)|^4)=(2m-1)!!\sum_{\xi\in\text{PMP}_{2m}}2^{\text{cyc}(\xi)},
\end{align}
where $\text{cyc}(\xi)$ is the number of cycles in $G_{\xi}(m)$. Now define the function
\beq
f(m):=\sum_{\sigma\in \text{PMP}_{2m}}2^{\text{cyc}(\xi)}
\eeq
so that
\beq\label{Eq: haf f(m)}
\mathbb{E}_{\mathcal{G}}(|\text{Haf}(X)|^4)=(2m-1)!!f(m).
\eeq

Our goal is to derive a recursion relation for $f(m)$. Focus on the nodes $(1,2)$. There are two possibilities for the action of $\xi$. One case is when $(12)$ appears in $\xi$. Call all such permutations $\xi_{12}$. For these permutations, the nodes $(1,2)$ already form a cycle in $G_{\xi}(m)$ and so it holds that
\begin{align*}
\sum_{\xi_{12}\in\text{PMP}_{2m}}2^{\text{cyc}(\xi)}&=2\sum_{\xi\in\text{PMP}_{2(m-1)}}2^{\text{cyc}(\xi)}\\
&=2f(m-1).
\end{align*}
Similarly, for all other permutations $\xi$ with an element of the form $(1j)$, for $j\neq 2$, we can combine the nodes $1,2$ and $j$ into a single node without changing the number of cycles in the graph. This is shown in Fig. \ref{Fig: EquivalentGraph}. Call these permutations $\xi_j$. Since there are $2m-2$ possible values of $j$, we have
\begin{align*}
\sum_{\xi_{j}\in\text{PMP}_{2m}}2^{\text{cyc}(\xi)}=& (2m-2)\sum_{\xi\in\text{PMP}_{2(m-1)}}2^{\text{cyc}(\xi)}\\
=&(2m-2)f(m-1).
\end{align*}
Combining these results we obtain the desired recursion relation $f(m)=2mf(m-1)$, which implies $f(m)=(2m)!!$. Setting $k=2m$, we finally obtain the desired expression for the second moment
\beq
\mathbb{E}_{\mathcal{G}}(|\text{Haf}(X)|^4)=(k-1)!!k!!=k!
\eeq

\end{document}